%% file: main.tex
\documentclass[letter, 10pt, conference]{ieeeconf}      
\IEEEoverridecommandlockouts                              
\input{preamble}
\input{macros}

\title{\LARGE \bf
Synthesis of Run-To-Completion Controllers for Discrete Event Systems 
}

\author{Yehia Abd Alrahman$^{1}$, Victor Braberman$^{2,3}$, Nicol\'as D'Ippolito$^{2,3}$, Nir Piterman$^{1,4}$ and Sebastian Uchitel$^{2,3,5}$%
\thanks{$^{1}$Department of Computer Science and Engineering, University of Gothenburg, Sweden.}%
\thanks{$^{2}$Departamento de Computac\'on, Facultad de Ciencias Exactas y Naturales, Universidad de Buenos Aires, Argentina.}%
\thanks{$^{3}$Instituto de Ciencias de la Computaci\'on, CONICET, Argentina.}%
\thanks{$^{4}$University of Leicester, Leicester, UK.}%
\thanks{$^{5}$Imperial College London, UK.}%
\thanks{This work is supported by the following grants: the ERC Consolidator grant D-SynMA (No. 772459), the Marie Skłodowska-Curie BeHAPI (No. 778233), the grants  ANPCYT PICT 2018-3835, ANPCYT PICT 2015-1718, CONICET PIP 2014/16 N°11220130100688CO, UBACYT 20020170100419 BA,  and the Swedish research council grants (No. 2020-04963) and  SynTM (No. 2020-03401).}
}

\begin{document}
\maketitle
\thispagestyle{empty}
\pagestyle{empty}

\input{tex/abstract}
\input{tex/intro}

\input{tex/background}
\input{tex/example}

\input{tex/problem}
\input{tex/analysis}

\input{tex/relatedWork}

\bibliographystyle{IEEEtran}
\bibliography{IEEEabrv,bibliography}
\clearpage
\onecolumn
\input{tex/appendix}
\end{document}

%% file: preamble.tex
\usepackage[utf8]{inputenc}
\usepackage{graphicx}
\usepackage{alltt}
\usepackage{amsfonts}
\usepackage{amssymb}
\usepackage{xspace}
\usepackage[pstricks1-10]{vaucanson-g}
\usepackage[nomargin,inline,index]{fixme}
\usepackage{subfigure}
\usepackage{url}
\usepackage{underscore}
\usepackage{xspace}
\usepackage{tikz}
\usetikzlibrary{positioning,shapes,shadows,arrows}
\usepackage[colorinlistoftodos, textwidth=40mm, textsize=footnotesize]{todonotes}
\usepackage{listings}
\usepackage{stmaryrd}
\usepackage{MnSymbol} 

%% file: macros.tex
\newcommand{\true}{\emph{true}\xspace}
\newcommand{\false}{\emph{false}\xspace}
\renewcommand{\implies}{\rightarrow}
\newcommand{\assume}{a}
\newcommand{\guarantee}{g}
\newcommand{\invariant}{\rho}
\newcommand{\fluent}{\ensuremath{f}\xspace}
\newcommand{\initf}{\emph{Init}_{\emph{\fluent}} }
\newcommand{\set}[1]{\ensuremath{\{#1\}}}
\newcommand{\gr}{{GR(1)}\xspace}
\newcommand{\sgr}{{SGR(1)}\xspace}

\newcommand{\alphabet}{\ensuremath{A}}
\newcommand{\controlledA}[1]{\ensuremath{\controlled{\alphabet}{#1}}}
\newcommand{\controlled}[2]{\ensuremath{{#1}_{#2}}}
\newcommand{\uncontrolledA}[1]{\ensuremath{\uncontrolled{\alphabet}{#1}}}
\newcommand{\uncontrolled}[2]{\ensuremath{\overline{\controlled{#1}{#2}}}}
\newcommand{\Uncontrollable}{\ensuremath{U}\xspace}
\newcommand{\Controllable}{\ensuremath{C}\xspace}

\newcommand{\livelockRemove}[1]{live({#1})}
\newcommand{\yieldE}{\gamma_{E}}
\newcommand{\yieldC}{\gamma_{C}}
\newcommand{\fairenv}{\ensuremath{\psi_{\mbox{\scriptsize e}}}}
\newcommand{\fairsys}{\ensuremath{\psi_{\mbox{\scriptsize c}}}}
\newcommand{\enenv}{\ensuremath{\mbox{en}_{e}}\xspace}
\newcommand{\ensys}{\ensuremath{\mbox{en}_{m}}\xspace}

\newcommand{\transition}[1]{\stackrel{#1}{\longrightarrow}}

\newcommand{\transitionM}[2]{{#1}^p_{_{#2}}}

\newcommand{\F}{\Diamond} 
\newcommand{\G}{\Box}

\newcommand{\U}{~\emph{U}~}
\newcommand{\W}{~\emph{W}~}
\newcommand{\actF}[1]{\overset{\centerdot}{#1}}

\newtheorem{theorem}{Theorem}[section]

\newtheorem{lemma}{Lemma}[section]
\newtheorem{corollary}{Corollary}[section]

\newtheorem{definition}{Definition}[section]

%% file: tex/abstract.tex
\begin{abstract}
A controller for a Discrete Event System must achieve its goals despite its environment being capable of resolving race conditions between controlled and uncontrolled events.
Assuming that the controller loses all races is sometimes unrealistic.
In many cases, a realistic assumption is that the controller sometimes wins races and is fast enough to perform multiple actions without being interrupted.
However, in order to model this scenario using control of DES requires introducing foreign assumptions about scheduling, that are hard to figure out correctly.
We propose a more balanced control problem, named run-to-completion (RTC), to alleviate this issue. 
RTC naturally supports an execution assumption in which both the controller and the environment are guaranteed to initiate and perform sequences of actions, without flooding or delaying each other indefinitely.
We consider control of DES in the context where specifications are given in the form of linear temporal logic.
We formalize the RTC control problem and show how it can be reduced to a standard control problem. 
\end{abstract}

%

%% file: tex/intro.tex
\section{Introduction}

The field of controller synthesis covers a spectrum of control problems, including \emph{Reactive Synthesis}~\cite{Pnueli:1989:SRM:75277.75293} and \emph{Supervisory Control}~\cite{ramadge89}. It targets dynamical systems whose state change is governed by the occurrence of discrete events. 
In these settings, system goals and the environment (or the uncontrolled plant) are specified as an accepted formal language, and the automatic synthesis procedure generates a correct-by-construction controller (or a supervisor). 
    
The controller must achieve its goals by dynamically disabling some of the controllable events based on the events that it has observed so far. The controller must be \emph{robust}. That is, it must be able to achieve its goals no matter what the environment does. However, the controller has no means of forcing the environment to generate an event. Thus, the environment not only identifies the possible controllable events in a given environment state, but also gets to choose the next scheduled event out of those selected by the controller and all enabled uncontrollable events. 
 
This asymmetric interaction between the environment and the controller represents a worst-case scenario that asks for producing robust controllers, achieving their goals despite the advantage offered to the environment.

In many application domains, this asymmetric interaction is too adversarial and requires adding explicit foreign assumptions that restrict the behaviour of the environment. 
One such application domain is that of embedded systems design in which reactive languages (e.g.,~\cite{Halbwachs10, DBLP:conf/ifip/Berry89, Benveniste03}) adopt a synchronous hypothesis where the system can react to an external stimulus with all the computation steps it needs~\cite{DBLP:reference/crc/SimoneTP05}. 
To handle such applications, the modeller is forced to introduce assumptions about the scheduling of the environment and the controller. These assumptions are not only hard to figure out correctly, but are also far from the actual focus of the control problem under consideration. In many cases, this may lead to generating controllers that satisfy their goals trivially by cornering the environment and disrupting its behaviour dramatically. Furthermore, the  written specifications become harder to read and understand, and consequently trickier to be incrementally developed due to their extensive dependencies.  


%
%
In this paper, we introduce a novel control problem, named \textit{run-to-completion} (RTC),  to mitigate the shortcomings of classical control, for such applications. RTC is a more balanced control problem that supports more natural modelling of systems that can initiate and perform sequences of actions in response to external stimuli. In essence, RTC provides a natural execution assumption in which both the controller and the environment may initiate and perform sequences of computation steps, without flooding or delaying each other indefinitely. 
Namely,
the controller has the ability to block environment actions for a finite time, this is akin to the controller stating that it still has something to do.
However, when the controller yields control back to the environment, it has to yield completely, i.e, the controller must allow all uncontrollable actions enabled by the environment. Furthermore, to support environment's run-to-completion, the controller must not interrupt the environment during its turn. 

RTC is suitable to control componentized systems where a response to a single external stimulus may require communication among subsystems. Due to flexible deadlines in RTC control, we are no longer required to count (or hardcode) the number of computation steps  for the system before it is ready again to react to the next stimulus. 

We show how to reduce RTC control to a modified control problem (i.e., with an asymmetric interaction). Furthermore, we show that RTC Control when used with  GR(1)~\cite{piterman06} goals can be reduced to Streett control of index 2~\cite{streettcondition1982}.

This paper is organised as follows: In Sect.~\ref{section:background} we present the necessary background and in Sect.~\ref{sec:example} we present a motivating example about a UAV reconnaissance mission. 
In Sect.~\ref{section:problem statement} and Sect.~\ref{section:examplerevistied}, we formally define RTC control and use the example to show its novel features. In Sect.~\ref{section:analysis}, we solve RTC control by a reduction to standard control. Finally, In Sect.~\ref{section:relatedwork}, we conclude and discuss related work.

%% file: tex/background.tex
\section{Background}
\label{section:background}

\subsection{Doubly-Labelled Transition System (DLTS)}
LTSs have been widely used for modelling and analysing the behaviour of concurrent and distributed systems (e.g.~\cite{magee06}. 
An LTS is a transition system where transitions are labelled with actions or events.
Here, as a part of the reasoning, we also label the states of the transition system with propositions, representing the set of events (or actions) that can be enabled from a specific state. Therefore, we use a DLTS instead. The use of DLTS is only a technicality and will not impact on the type of the generated controllers. In fact, state labels will be ignored in the generated controllers.


\begin{definition}
\label{def:LTKS} \emph{(DLTS)} 
A DLTS is $T = (S, P, A,\Delta,$ $L, s_0)$, where $S$ is a finite set of states, $P$ is a set of \emph{state propositions}, $A$ is a {\em transition alphabet} partitioned $\alphabet = \controlledA{T} \cupplus \uncontrolledA{T}$ to actions controlled by $T$ and actions monitored by $T$, $\Delta \subseteq (S \times A \times S)$ is a transition relation,  $L:S \rightarrow 2^P$ is a labeling function, and $s_0 \in S$ is the initial state.  
\end{definition}

We denote
$\Delta_\ell(s) = \{s' ~|~ (s,\ell,s')\in \Delta\}$, $\Delta_{\alphabet'}(s)=\bigcup_{\ell\in\alphabet'}\Delta_\ell(s)$, and $\Delta(s)=\Delta_\alphabet(s)$.
This notation is extended to sets of states, e.g., $\Delta_\ell(S') = \bigcup_{s\in S'}\Delta_\ell(s)$.
We say $\ell$ is enabled in state $s$ if $\Delta_\ell(s)\neq \emptyset$.

We say a DLTS is transition-deterministic if $(s,\ell,s')$ and $(s,\ell,s'')$ are in $\Delta$ implies $s'=s''$.
An execution of $T$ is a maximal sequence of states and transition labels $\pi = s_0, a_0, s_1, \ldots$ where $s_0$ is the initial state and for every $i\geq 0$ we have $(s_i, a_i, s_{i+1}) \in \Delta$. 


\begin{definition}\emph{(The Parallel Composition of DLTS(s))}\label{def:parallelcomp}
Let $M = (S_M, P_M, A_M, \Delta_M,  L_M, s_{0_M})$ and $E = (S_E, P_E, A_E, \Delta_E, L_E, s_{0_E})$ be two DLTSs.
The \emph{parallel composition} of $M$ and $E$ is defined by a symmetric and a binary operator $\|$ such that $M \| E$ is also a DLTS $T = (S_M\times S_E, P, A_M \cup A_E, \Delta, L, (s_{0_M},s_{0_E}))$, where $P=P_M\cupplus P_E$, $L(m,e)=L_M(m)\cupplus L_E(e)$, and 
$\Delta$ is the smallest relation that satisfies the rules below, 
\medskip

\begin{center}
\begin{tabular}{lcl}
    $\!  \frac{m \transition{\ell} m'} {(m,e) \transition{\ell} (m',e)} \, { \small \ell \, \not\in \, A_E}$ & 
    \hspace{0.01in} &
    $\!  \frac{e \transition{\ell} e'} {(m,e) \transition{\ell} (m,e')} \, {\small \ell \, \not\in \, A_M}$ \\
    \multicolumn{3}{c}{
    $\! \frac{m \transition{\ell} m', \ e \transition{\ell} e'} {(m,e) \transition{\ell} (m',e')} \, {\small \ell \, \in \, A_M\cap A_E}$ }

\end{tabular}
\end{center}
\end{definition}\medskip

Note that parallel composition only synchronise on actions, and thus preserves the proposition values of the two separate parts.
 
\subsection{Fluent Linear Temporal Logics}
Linear temporal logics are widely used to describe and analyse behaviour requirements~\cite{gianna03,vanLamsweerde:2000:HOG:357525.357521,Kazhamiakin:2004:FVR:1030033.1030074, Uchitel04}. 
The \emph{Fluent Linear Temporal logic (FLTL)}~\cite{gianna03} replaces state propositions in traditional temporal logics with fluents. A fluent is a predicate over a set of initiating and terminating actions. Once triggered by an initiating action, a fluent continues to hold as long as no terminating action is enabled.  
Thus, FLTL provides a uniform framework for specifying both instantaneous actions and also actions that take time~\cite{gianna03,Letier05}. 
To simplify notations we do not include the next operator in FLTL. 
All our results can be easily generalised to include the next operator.

FLTL was designed for LTS, here we adapt it for DLTS by introducing fluents to also account for the propositions that label states of a DLTS. We introduce two types of fluents: \emph{transition fluents} and \emph{proposition fluents}.
A \emph{transition fluent} \fluent is defined by a pair of sets of actions and a Boolean value: $\emph{\fluent} = \langle I_{\emph{\fluent}}, T_{\emph{\fluent}}, \initf \rangle$, where $I_{\emph{\fluent}}\subseteq Act$ is the set of initiating actions, $T_{\emph{\fluent}}\subseteq Act$ is the set of terminating actions and $I_{\emph{\fluent}}\cap T_{\emph{\fluent}}=\emptyset$. 
A transition fluent may be initially \true or \false as indicated by \emph{Init}$_{\emph{\fluent}}$.  
Every action $\ell\in Act$ induces a \emph{transition fluent}, namely $\actF{\ell}=\langle \ell, Act\setminus \set{\ell},\ \false\rangle$.
Every state proposition $p$ of a DLTS induces a \emph{proposition fluent} $p$.

Let $\mathcal{F}$ be the set of all fluents over $Act$ and $P$. 
An FLTL formula is built up from the standard Boolean connectives and the temporal 
operator $\U$ (strong until) as follows: 
\begin{eqnarray*} 
\varphi ::= \fluent \mid \neg \varphi \mid \varphi \vee \psi \mid \varphi \U \psi, 
\end{eqnarray*}
where $\fluent\in\mathcal{F}$. 
As usual we introduce $\wedge$, $\F$ (eventually), $\G$ (always), and $\W$ (weak until) as syntactic sugar. 
Let $\Pi$ be the set of infinite executions of a DLTS $T$ over \emph{Act} and $P$.
For an execution $\pi=s_0,\ell_0,s_1,\ell_1,\ldots$, we say it satisfies a transition fluent $\emph{f}$ at position $i$, denoted $\pi,i \models \emph{f}$, if and only if one of the following conditions holds:
\begin{list}{-}{\leftmargin=2em}
\item $\initf \wedge (\forall j \in \mathbb{N} \cdot 0 \leq j \leq i \rightarrow \ell_j \notin T_{\fluent})$
\item $\exists j \in \mathbb{N} \cdot (j \leq i \wedge \ell_j \in I_{\fluent}) \wedge (\forall k \in \mathbb{N} \cdot j < k \leq i \ \rightarrow \ell_k \notin T_{\fluent})$
\end{list}
It satisfies a proposition-fluent $p$ at position $i$, denoted $\pi,i\models p$, if and only if $p\in L(s_i)$.

Given an infinite execution $\pi$, the satisfaction of a formula $\varphi$ at position $i$, denoted $\pi,i\models\varphi$, is defined as follows: 
\[
\begin{array}{lcl}
\pi,i \models \fluent &\triangleq & \pi,i \models \fluent  \\
\pi,i \models \neg \varphi & \triangleq & \neg(\pi,i \models \varphi) \\
\pi,i \models \varphi \vee \psi &\triangleq& (\pi,i\models \varphi) \vee (\pi,i \models \psi) \\
\pi,i \models \varphi \U \psi &\triangleq&\exists j \geq i \cdot \pi,j \models \psi \wedge \forall \mbox{ }i \leq k < j \cdot \pi,k \models \varphi\\
\end{array}
\]

We say that $\varphi$ holds in $\pi$, denoted $\pi\models\varphi$, if $\pi,0\models\varphi$. 
A formula $\varphi \in \mbox{FLTL}$ holds in an LTS $T$ (denoted $T \models \varphi$) if it holds on every infinite execution produced by $T$.

We assume that user supplied specifications do not use proposition fluents.
However, proposition fluents are required for various parts of our analysis. 



\subsection{Controller Synthesis}
The standard control problem is as follows: 
Consider an FLTL formula $\varphi$ and a DLTS model $E$ of the environment, with the set of actions $\alphabet$ partitioned into 
environment actions $\controlledA{E}$ and monitored controller actions $\uncontrolledA{E}$.
Construct a DLTS $M$ to control $\uncontrolledA{E}$ and to monitor $\controlledA{E}$ such that when composed with $E$ (i.e. $E\|M$), the controller does not block environment actions (i.e. actions in $\controlledA{E}$),  $E\|M$ is deadlock-free, and every execution of $E\|M$ satisfies $\varphi$.
For simplicity (and taking the controller's point of view), we uniformly denote by $\Uncontrollable$ (for uncontrollable) the set $\controlledA{E}=\uncontrolledA{M}$ and by $\Controllable$ (for controllable) the set $\uncontrolledA{E}=\controlledA{M}$.
That is, $U$ is the set of actions controlled by the environment and monitored by the controller and $C$ is the set of actions monitored by the environment and controlled by the controller. 

A legal controller does not block the actions in $\Uncontrollable$ and enables only actions in $\Controllable$ that are available.  
This notion is based on that of {\em legal environment} for Interface Automata~\cite{alfaro01}. 
Formally \emph{legality} is defined as follows.

\begin{definition}\label{def:legalEnvironment}\emph{(Legality)}
Consider a DLTS model of the environment $E = (S_E, P_E , \alphabet, \Delta$, $L_E, s_{E_0})$ and a DLTS model of the controller $M = (S_M, P_M, \alphabet,\Gamma,L_M,s_{M_0})$, where $\alphabet=\Uncontrollable \cupplus \Controllable$. 
We say that $M$ is legal for $E$ if for every reachable state $(m,e)$ of $M\| E$ the following holds.
\begin{itemize}
    \item For all $\ell \in \Uncontrollable$ such that $\Delta_\ell(e)\neq \emptyset$ we have $\Gamma_\ell(m)\neq \emptyset$.
    \item For all $\ell \in \Controllable$ such that $\Delta_\ell(e)=\emptyset$ we have $\Gamma_\ell(m)=\emptyset$.
\end{itemize}
\end{definition}


\begin{definition}\label{def:lts-control}\emph{(Standard Control)}
Given a domain model in the form of a DLTS $E=(S, P, \alphabet, \Delta,  L, s_{0})$, where $\alphabet = \Uncontrollable\cupplus\Controllable$, and an FLTL formula $\varphi$, a solution for the DLTS control problem $\mathcal{E}=\langle E, \varphi, \Controllable\rangle$ is a DLTS $M=(S_M, P_M, A, \Delta_M, L_M, s_{0_M})$ such that $M$ is legal for $E$, $E\|M$ is deadlock free, and $E\|M\models \varphi$.  
\end{definition}

The synthesis problem for FLTL  is 2EXPTIME-complete~\cite{dippolito13phd}. 
Nevertheless, restrictions on the form of the goal and assumption specifications have been studied and found to be solvable in polynomial time. 
For example, goal specifications consisting uniquely of safety requirements can be solved in linear time, and particular styles of liveness properties such as \gr~\cite{piterman12} can be solved in quadratic time. 
An adaptation of \gr in the context of LTS has been presented in~\cite{dippolito13} and is defined as follows:

\begin{definition}\label{def:sgr-lts-control}\emph{(\sgr DLTS Control)}
A DLTS control problem $\mathcal{E}=\langle E, \varphi, \uncontrolledA{E}\rangle$ is \sgr if $E$ is deterministic, and $\varphi$ is of the form $\varphi=\G \invariant \wedge (\bigwedge_{i=1}^n \G\F \assume_i \rightarrow \bigwedge_{j=1}^m \G\F \guarantee_j)$, where $\invariant$, $\assume_i$ and $\guarantee_j$ are Boolean combinations of fluents. Note that $\G \invariant$ is a safety condition on both the environment and the controller. Furthermore $\assume_i$ and $\guarantee_j$ are liveness assumptions and guarantees on the environment and the controller respectively.  
\end{definition}

%% file: tex/example.tex
\section{Motivating Example}
\label{sec:example}
Consider a reconnaissance mission for a UAV, surveying a discretised area. The UAV controls the following actions:  $\mathsf{takeoff, go[x][y], takePicture[x][y], econoMode,}$ and $\mathsf{land}$. However, during surveillance, the UAV is required to monitor environment actions: $\mathsf{arrive[x][y], lowBat,}$ and $\mathsf{criticalBat}$. 

The behaviour exhibited by the UAV when not controlled is depicted on the left of Fig.~\ref{fig:UAVENV}, where after taking off it may do an arbitrary action (except for $\mathsf{takeoff}$ and $\mathsf{land}$) in its alphabet $A$ before it finally lands. The safety assumption on the environment, as depicted on the right of Fig.~\ref{fig:UAVENV}, ensures that $\mathsf{arrive[x][y]}$ may only happen as a result of a $\mathsf{go[x][y]}$ action. For the sake of presentation, we only consider an area, consisting of two locations: (1,1) and (1,2).


\begin{figure}[h]
\begin{tabular}{cc}
\includegraphics[scale=.2]{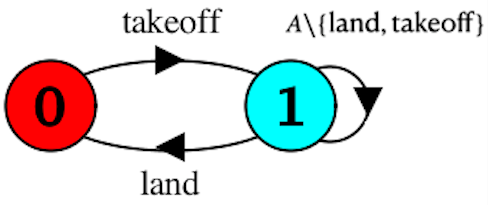}&\quad
\includegraphics[scale=.35]{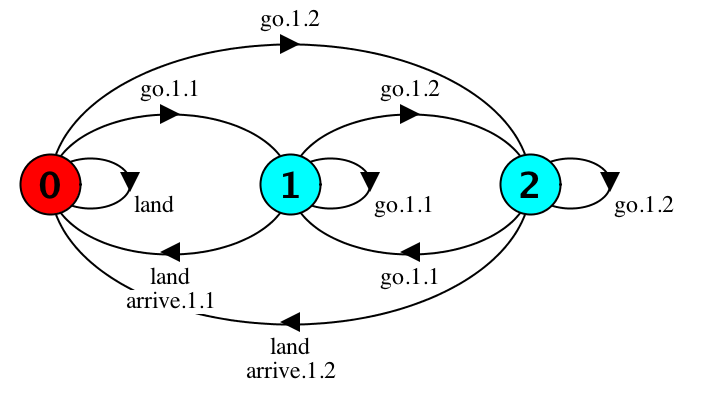}
\\
\end{tabular}
\caption{(left) Model of the UAV  and (right) the Environment assumption}
\label{fig:UAVENV}
\end{figure}


We want to synthesise a controller for the UAV, satisfying the following safety goals: 
\begin{enumerate}

\item Landing must only occur after taking a picture for every locations or upon a critical battery alert.
\[ \G (\actF{land} \implies  \forall x,y \ \cdot \ Sensed[x][y] \ \vee \ CritBat)\]
where fluent $Sensed[x][y]$ is defined as $\langle\{\mathsf{takePicture[x][y]\}, \{takeoff\}},\false\rangle$ and fluent $CritBat$ as $\langle\mathsf{\{criticalBat\}, \{takeoff\},\ \false}\rangle$

\item Taking a picture for a particular location must only happen at that location:
\[\G (\forall x,y \cdot \actF{takePicture[x][y]} \implies At[x][y])\]

 where fluent $At[x][y]$ is defined as $\langle\{\mathsf{arrive[x][y]}\},$ $\mathsf{\{go[x'][y'], land\},\ \false}\rangle$

 \item Low battery alerts must trigger economy flying mode as soon as possible: 
\[
\begin{array}{c}
\G (\actF{lowBat} \implies ((\neg \bigvee_{\ell \in C\setminus \{ \actF{land},\actF{econoMode} \}} \ell) \hspace{1cm}\\[1ex]
\hfill\W \actF{economode})
\end{array}
 \]

 \item Critical battery alerts must trigger immediate landing: 
\[
\begin{array}{c}
\G (\actF{criticalBat} \implies ((\neg \bigvee_{\ell \in C\setminus\{\actF{land}\}} \ell) \W \actF{land}))
\end{array}
 \]

\end{enumerate}

Finally, the liveness goal for the UAV controller is always eventually landing: $\G\F land$. We stress that 
the safety of landing implies that this happens only after having completed the survey or in response to a critical battery alert. Furthermore, when the UAV issues a $\mathsf{go[x][y]}$ command, we require that the environment ensures always eventual arrival: $\G\F \neg PendingArrival$, where fluent $PendingArrival$ is defined as $\langle\{\mathsf{go[x][y]}\}, \{\mathsf{arrive[x][y], land}\},\ \false\rangle$.

No solution for this control problem exists because the environment can flood the controller by generating an infinite number of $\mathsf{lowBat}$ and $\mathsf{criticalBat}$ events, impeding all controlled actions and hence progress towards the liveness goal. The non existence of such a solution stems from an unrealistic assumption on the environment. Namely, that the environment may impede the progress of the controller merely by a continual notification of a drained battery.  

A natural environment assumption that can be introduced to avoid this is to cap the number of $\mathsf{lowBat}$ and $\mathsf{criticalBat}$ events. In  Fig.~\ref{fig:NaiveNoFlood} we show one such constraint in which a maximum of one $\mathsf{lowBat}$ and one $\mathsf{criticalBat}$ can occur between  $\mathsf{go}$ commands.

However, this assumption yields controllers that once they have taken off they keep hovering until their batteries are drained, and consequently they land. This way they meet all of their safety goals while achieving their liveness goals by cornering the environment and restricting its set of possible actions to $\mathsf{lowBat}$ and $\mathsf{criticalBat}$, i.e, they never issue $\mathsf{go}$ commands. Note that if $\mathsf{go}$ commands are never issued, $\mathsf{arrive}$ events cannot occur and thus the only environment events that can and will eventually occur  are $\mathsf{lowBat}$ and $\mathsf{criticalBat}$.  
Indeed, such controllers would never allow the UAV to complete the surveying mission (i.e., landing always occurs because of $criticalBat$ and never because of having achieved $Sensed[x][y]$ for all $x$ and $y$).

\begin{figure}[h]
\centering
\includegraphics[scale=.35]{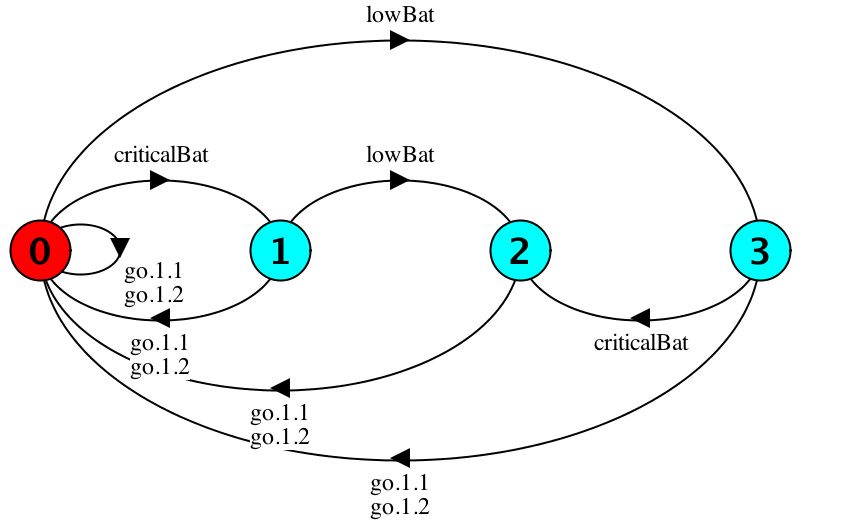}
\caption{Naive Environment Assumption to Avoid Flooding.}
\label{fig:NaiveNoFlood}
\end{figure}

The assumption that does the trick while avoiding to synthesise trivial controllers is achieved by restricting  $\mathsf{lowBat}$
and 
$\mathsf{criticalBat}$
to happen \emph{once and only} after issuing a $\mathsf{go[x][y]}$ command (i.e. when $\mathsf{arrived}$ events are also enabled). This assumption is depicted in Fig.~\ref{fig:NoFlood}. 

\begin{figure}[h]
\centering 
\includegraphics[scale=.35]{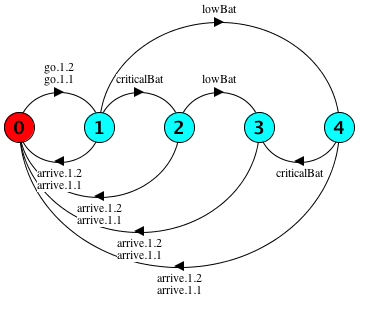}
\caption{Assumption to Avoid Flooding.}
\label{fig:NoFlood}
\end{figure}

The example clearly shows how simplifying assumptions (e.g., Figures~\ref{fig:NaiveNoFlood} and \ref{fig:NoFlood}) can be tricky for the modeller to figure out, while avoiding unrealistic situations for unrealizability in control problems and also avoiding trivial solutions.



%% file: tex/problem.tex
\section{Problem Statement}
\label{section:problem statement}

As exhibited in the previous section, the prevalent approach to control in Discrete Event Systems poses serious modelling problems related to the continual triggering of environment events that have to be dealt with by the controller. In many applications, the controller has to execute a sequence of steps (or a finite protocol) in response to a single event. However, this would not be possible if the environment keeps triggering events, flooding the controller with uncontrollable events, and thus impeding its progress towards completing its designated tasks.

By definition, the environment has a double role.
It identifies the possible controller actions in a given environment state and it also represents the (adversarial) behaviours.
The first corresponds to physical/software restrictions (e.g., $\mathsf{go}$ happens only after $\mathsf{takeoff}$); the second corresponds to the scheduling of the next event (e.g., $\mathsf{criticalBat}$ can always win the race against $\mathsf{land}$), because the controller cannot disable environment events and the environment always picks the next event out of those selected by the controller and all enabled uncontrollable events in an environment state.
 
To mitigate this problem, the modeller has to carefully consider how to restrict the environment.
Essentially, the modeller is forced to introduce assumptions about the scheduling of the environment and the controller that are an artifact resulting from the definition of (an asymmetric) control problem. These assumptions are not only hard to figure out correctly, but are also far from the actual focus of the control problem under consideration. This makes written specifications harder to read and understand, and consequently trickier to be incrementally developed due to their extensive dependencies.   
Here, we suggest a more balanced control problem. 
We call this approach \emph{Run-to-Completion (RTC)} as both 
the environment and the controller can perform sequences of actions.
At the same time, neither can flood the other or delay it indefinitely.


In RTC control the notion of legality is more subtle. 
The controller has to be able to disable environment actions, this is akin to the controller stating that it still has something to do.
However, when some uncontrollable action is enabled by the controller, the controller must allow all of them. Furthermore, to support environment's run-to-completion, the controller must not interrupt the environment when it is the environment that is moving. This amounts to saying that if the environment has moved, the controller must enable all uncontrollable actions.
This is formalized below. 
\begin{definition}\label{def:rtclegalEnvironment}\emph{(Legality under RTC semantics)}
Consider the DLTSs $E = (S_E, P_E, \alphabet, \Delta,  L_E, s_{E_0})$ and  $M = (S_M, P_M, \alphabet,\Gamma,L_M,s_{M_0})$, where $\alphabet=\Uncontrollable \cupplus \Controllable$.
We say that $M$ is run-to-completion (RTC) legal for $E$ if for every reachable state $(e,m)$ of $E\| M$ the following holds.
\begin{itemize}
\item 
When allowing the environment to move, allow all its possible actions:
If $\Gamma_{\Uncontrollable}(m)\neq\emptyset$  then for every $\ell\in \Uncontrollable$ such that $\Delta_{\ell}(e)\neq \emptyset$ we have that $\Gamma_{\ell}(m)\neq \emptyset$.
\item
After uncontrolled actions, let the environment progress towards completion: If $m\in \Gamma_{\Uncontrollable}(S_M)$ then for every $\ell\in \Uncontrollable$ such that $\Delta_{\ell}(e)\neq \emptyset$ we have that $\Gamma_\ell(m)\neq \emptyset$.
\item
 For every $\ell\in\Controllable$ such that $\Delta_\ell(e)=\emptyset$ we  have that $\Gamma_\ell(m)=\emptyset$.
\end{itemize}
\end{definition}

Additionally, we have to ensure that both the environment and the controller are non-Zeno.
That is, both do not take an infinite sequence of actions without giving the other opportunities for making progress. 
{On the controller side, we require that all computations are (controller) non-Zeno. 
On the environment side, we consider only (environment) non-Zeno computations for the satisfaction of the goal. The latter is because it is valid for a controller to chose never to take a controlled action if this ensures its goal.  }

To formalize the non-Zeno assumption we first introduce four auxiliary formulas: $c$, $u$, $pass_E$, and $pass_M$.
Given an environment to be controlled $E$, a candidate controller $M$, and their parallel composition $E\parallel M$,
we assume that in both $E$ and $M$ (separately) for every $\ell\in\alphabet$ there are propositions $\transitionM{\ell}{E} \in P_E$ and $\transitionM{\ell}{M}\in P_M$ such that $(s,\ell,s')\in \Delta_E$ iff $\transitionM{\ell}{E} \in L(s)$ and similarly for $M$. 
Let 
$c=\bigvee_{\ell\in\Controllable}\actF{\ell}$,
$u=\bigvee_{\ell\in\Uncontrollable}\actF{\ell}$,
$pass_M=\bigwedge_{\ell\in\Controllable} \neg \transitionM{\ell}{M}$,
and let
$pass_E=\bigwedge_{\ell\in\Uncontrollable} \neg \transitionM{\ell}{E}$.
That is, $c$ and $u$ are formulas specifying the possibility of executing some controllable and uncontrollable actions, respectively. 
The formulas $pass_E$ and $pass_M$ characterize states where the environment and, respectively, the controller do not enable any uncontrollable and controllable action. 
That is, in $pass_E$ all uncontrollable actions are \emph{impossible} in the environment and in $pass_M$ all controllable actions (if exist) are \emph{not enabled} by the controller.
Note that by the definition of legality the controller can only enable controllable actions that are enabled in the environment. 

We now define the formulas $\fairenv$ and $\fairsys$ denoting non-Zeno-ness assumptions on the environment and the controller respectively.
Let $\fairenv = \G\F (c \ \vee \ pass_M)$.
That is (if enabled in the environment model), the environment allows infinitely many controllable actions (by $M$) in the execution or there are infinitely many states visited in which $M$ does not enable controllable actions.
Let $\fairsys = \G\F (u \ \vee \ pass_E)$.
That is, the controller allows infinitely many uncontrollable actions (by $E$) in the execution or there are infinitely many states visited in which $E$ does not enable uncontrollable actions.

\begin{definition}\label{def:zeno-rtc-lts-control}\emph{(RTC  Control)}
Given an environment model $E=(S, P_E, \alphabet, \Delta,  L_E, s_{0})$ and an FLTL formula $\varphi$, where $\alphabet = \Uncontrollable\cupplus \Controllable$ is defined as before. 
A solution for the RTC control problem $\mathcal{E}=\langle E, \varphi, \Controllable\rangle$ is a DLTS $M=(S_M$, $P_M$, $\alphabet$, $\Delta_M$,  $L_M$, $s_{0_M})$ such that $M$ is RTC legal for $E$, 
$E\parallel M$ is deadlock free, and every execution 
$\pi$ of $E\|M$ satisfies $\pi\models \fairsys \wedge (\fairenv 
\implies \varphi)$.  
\end{definition}

We note that we cannot move the condition $\varphi_c$ into the implication. Indeed, this would imply that by \emph{not} fulfilling $\varphi_c$ the controller is able to force the environment to violate $\varphi_e$ as well. Thus, the controller would trivially fulfil the goal by blocking the environment forever. 

\section{Example Revisited}\label{section:examplerevistied}
We revisit the example in Sect.~\ref{sec:example} under RTC control. We show how RTC control relieves the modeller from dealing with intricate scheduling issues that are hard to figure out correctly. In fact, the modeller is no longer required to come up with foreign modelling artefacts (i.e., Fig.~\ref{fig:NaiveNoFlood} and Fig.~\ref{fig:NoFlood}) to avoid unrealistic situations for unrealizability or to ensure run-to-completion. We also show how RTC control permits writing loosely-coupled specifications, and thus facilitates incremental development.


In Fig.~\ref{fig:RTCcontrollerTwoStates}, we show a snippet of the RTC controller for the surveillance mission. Due to the change of control mode we no longer need to cap the number of uncontrolled events to avoid flooding. This is naturally captured in RTC control as both the environment and the controller may perform finite sequences  of actions without flooding or delaying each other indefinitely.
Note the path via states 0, 1, 2, 3, 17, 16 where the UAV has arrived to location [1][1] but both the $\mathsf{criticalBat}$ and the $\mathsf{lowBat}$ alarms have been raised. In state 16, an arbitrary number of uncontrolled events (i.e., $\mathsf{arrive[1][1], criticalBat, lowBat}$) can occur; however if fairness assumption $\psi_e$ holds, then the controller will get a chance to execute and at this point it can perform all the controlled actions it needs to do in order to satisfy the \textit{safety} requirements (3) and (4) in Sect.~\ref{sec:example} (i.e., $\mathsf{land}$ and $\mathsf{econoMode}$ via state 13 to reach state 0). 

The same path shows how the environment can also run to completion, raising $\mathsf{lowBat}$ and $\mathsf{criticalBat}$ (states 3, 17 and 16) should it want to. Alternatively,  it may forfeit its turn and the controller may land from  state 17. 

Also note how in state 3, if given the chance, the controller will both take the picture it needs and go to the next location (states 6 and 7) even though there is no explicit  requirement. It does so, as it attempts to do as many actions as it can while progressing towards its liveness goal (i.e., land). 


\begin{figure}[ht]
\includegraphics[width=0.45
\textwidth]{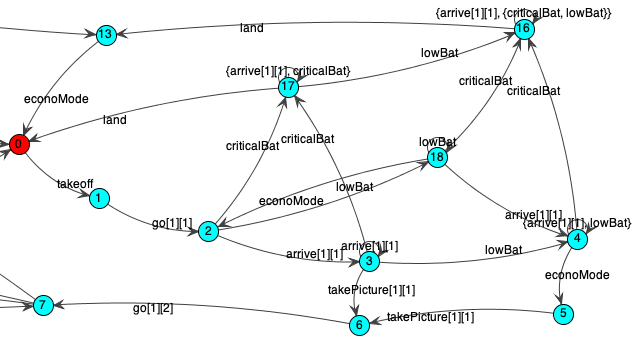}
\caption{RTC Controller (Snippet)}
\label{fig:RTCcontrollerTwoStates}
\end{figure}


Interestingly, the fairness assumptions $\fairenv$ and $\fairsys$ in RTC control remove any possible clashes among subgoals in the specifications, by maintaining fair executions and flexible deadlines. Consider, for instance, the coupling between the \textit{safety} requirements (or subgoals) (3) and (4) in Sect.~\ref{sec:example} to hardcode the concept of \emph{as soon as possible}. 

In standard control, the next controllable actions after a $\mathsf{lowBat}$ or a $\mathsf{criticalBat}$ in $(3)$ and $(4)$ cannot be specified independently as otherwise they would contradict each other, e.g., if the next controllable action after $\mathsf{lowBat}$ needs to be $\mathsf{econoMode}$ instead of $\mathsf{land},\mathsf{econoMode}$, then if $\mathsf{lowBat}$ and $\mathsf{criticalBat}$ happen consequtively, the controller cannot respond to either one. These subgoals require attention due to the subtle interactions between them.

In RTC control, we write specifications on top of a fair interaction model where both the controller and the environment are given the chance to run to completion, without counting or fixing deadlines. This suggests that we can write cleaner specifications that remove the explicit dependencies among subgoals, and thus enhancing the readability, and consequently facilitating incremental development.  

That said, we introduce convenient schemata to naturally specify high-level concepts like \textit{as soon as possible} and \textit{urgent response} under RTC control. We believe these schemata are more intuitive and less error prone. The schemata relate two Boolean combinations of fluents $\phi$ and $\psi$: 
$$
\begin{array}{l c l}
\mbox{\textsc{Asap}}(\psi)& = & ((\bigwedge_{\ell \in C} \neg \actF{\ell}) \ W ((\bigvee_{\ell \in C} \actF{\ell}) \ W \psi)) \\
\mbox{\textsc{urgRsp}}(\phi, \psi)& = & \G [\phi \implies \mbox{\textsc{Asap}}(\psi)]
\end{array}
$$


Now, we may replace subgoals (3) and (4) in Sect.~\ref{sec:example} with more natural enunciations. Namely, that $\mathsf{econoMode}$ and $\mathsf{land}$ are \textit{urgent} response requirements to $\mathsf{lowBat}$ and $\mathsf{criticalBat}$ respectively: 
\[{\textsc{urgRsp}}(\actF{lowBat},\actF{econoMode}) \wedge {\textsc{urgRsp}}(\actF{criticalBat},\actF{land})\]

Clearly, these schemata removed the coupling in (3) and (4), while maintaining correctness under RTC control.


%% file: tex/analysis.tex
\section{Analysis}
\label{section:analysis}

We show how to solve the RTC control problem by a reduction to a standard control problem. Recall that the environment in standard control  always gets to choose the next event out of those selected by the controller and all enabled uncontrollable events in a given environment state. Therefore, we need to model the ``act'' of yielding control explicitly in the modified control problem. Thus, the analysis may refer directly to when each side is yielding control. 


Consider an RTC control problem $\varepsilon = \langle E, \varphi, \Controllable\rangle$.
We now define a DLTS that captures the transference of control between the environment and the controller:

\begin{definition}\label{def:yieldLTS}\emph{(Yield DLTS)}
Let $\Uncontrollable$ and $\Controllable$ be the set of actions controlled by the environment and the controller, respectively. 
The yield DLTS is defined as $Y = (\set{c, e}, \emptyset, \set{\yieldC,\yieldE} \cup \Controllable \cup \Uncontrollable, \Delta,  L, e)$. Where 
$\Delta = \set{(c, \yieldC, e), (e, \yieldE, c)}
\cup \set{e \transition{\ell} e\ |\ \ell \in \Uncontrollable}
\cup \set{c \transition{\ell} c\ |\ \ell \in \Controllable}$.
\end{definition}

Now, we define a (standard) control problem over $E\parallel Y$. Let $\Uncontrollable^{+}$ be $\Uncontrollable \cup \{\yieldE\}$ and $\Controllable^{+}$ be $\Controllable \cup \{\yieldC\}$.
We use the fluents \enenv and \ensys, which indicate whether $Y$ is in state $e$ or $c$.
Formally, $\enenv = \langle\set{\yieldC},\set{\yieldE},\true\rangle$ and $\ensys = \langle\set{\yieldE},\set{\yieldC},\false\rangle$. 
Intuitively, this corresponds to RTC legality as when the environment $E\| Y$ is in a state of the form $(s,e)$ all uncontrollable actions are enabled. Furthermore, uncontrollable actions remain in states of this form (and thus cannot be interrupted).
As for controllable actions, they lead to states of the form $(s,c)$, where only controllable actions are enabled, allowing the controller to take a sequence of actions.

The composition $E\parallel Y$ turns all deadlocks in $E$, which the controller should avoid, to livelocks, where the two sides cooperate to stop time.  
That is, $E\parallel Y$ gets trapped in an infinite sequence of yield transitions (or Livelock cycle) $s_0,\gamma_E,s_1,\gamma_C,s_2,\gamma_E,\dots$, namely when both $\Delta_{C}(s_i)$ and $\Delta_{U}(s_i)$ are empty.

The Livelock Removal Operator, defined below, removes livelock cycles in a DLTS by removing yield transitions by the controller (resp. environment) to states in which the environment (resp. controller) can only yield back. 

\begin{definition}\label{def:removeLivelock}\emph{(Livelock Removal Operator)}
Let $N = (S, P, A, \Delta,  L, s_0)$ be a DLTS obtained by parallel composition with $Y$.
The \emph{livelock removal operator}  $\livelockRemove{N}$ is a DLTS $(S, P, A, \Delta', L, s_0)$ where 
\[\begin{array}{c}
\Delta' = \set{ (s,\ell,s') \in \Delta ~|~ \ell \in \set{\yieldC,\yieldE} \implies\hfill\\[1ex] \hfill\exists \ell' \notin\set{\yieldC,\yieldE} \cdot\ \Delta_{\ell'}(s')\neq \emptyset}\\[1ex] 
\end{array}
\]
\end{definition}

That is, the transitions with actions $\yieldC$ or $\yieldE$ are retained only if the environment/controller can do something other than yielding back control immediately. 
Following the removal of livelocks, we can reduce the RTC control to the following control problem.

\begin{theorem}\label{thm:analysis-control}\emph{(Analysis Control)} Consider 
an environment model $E=(S, P, \alphabet, \Delta,  L, s_{0})$,  
where $\alphabet=\Uncontrollable 
\cupplus \Controllable$ the set of actions controlled and monitored by $E$ respectively, and an FLTL formula $\varphi$.
A solution for the RTC control problem $\varepsilon=\langle E, \varphi, \Controllable\rangle$ exists if and only if the standard control problem $\varepsilon^+=\langle \livelockRemove{E\|Y}, \varphi^+, \Controllable^+\rangle$ is controllable, and $\varphi^+$ is defined below. For simplicity, we use $A$ to denote $\bigvee_{\ell \in A}\actF{\ell}$.
$$\begin{array}{l l}
\varphi^+ = & \G\F(\enenv \vee (\bigwedge_{\ell \in \Uncontrollable}\neg \ell^p_E)) ~\wedge \\
& (\G\F (\ensys \vee  (\bigwedge_{\ell \in \Controllable}\neg \ell^p_E))
\rightarrow [\G\F \alphabet  
\rightarrow  \varphi])\\[1ex]\end{array}$$  

\end{theorem}


Note that $\varphi^+$ mimics formula $\fairsys \wedge (\fairenv \implies \varphi)$ of the RTC control formulation with  $\fairsys = \G\F (u \ \vee \ pass_E)$
and $\fairenv = \G\F (c \ \vee \ pass_M)$. 
There are two differences.
First, $\varphi^+$ replaces $pass_M$ by $\bigwedge_{\ell \in\Controllable}\neg \ell^p_E$.
Second, $\varphi^+$ disregards traces in $\epsilon^+$ in which the environment and the controller collaborate to stop time. This is done by evaluating $\varphi$ only on traces satisfying $\G\F A$.
The second is not a problem when extracting an RTC controller as the environment of the RTC controller cannot stop time. 

The proof of Theorem~\ref{thm:analysis-control} (reported in the full version on~\cite{alrahman2020synthesis} due to space limitations) shows that given 
a solution $M^+$ of $\varepsilon^+$ we can construct an RTC 
solution $M$ of $\varepsilon$.
The size of $M$ is at most twice the size of $M^+$.

\begin{corollary} Given an DLTS solution $M^+$
    for the modified control problem $\varepsilon^+$, 
    we can construct a solution $M$ to $\varepsilon$ such that the number of states of $M$ is at most twice the number of states of $M^+$. Furthermore, if $M^+$ is deterministic then $M$ is also deterministic.

\end{corollary}
\begin{proof} We only report the construction of $M$ as its correctness is immediate from the proof of Theorem~\ref{thm:analysis-control}.
Let $M^+=(T, P, A^+, \Gamma^+,L^+,t_0^+)$. 
The components of $M$ are defined as follows:
\begin{itemize}
    \item 
    $T_M=\set{e,c} \times T$ is the set of states;
    \item The alphabet
    $A_M=A^+ \ \backslash\ \set{\yieldC, \yieldE}$;
    \item
    The initial state $t'_0=(e,t_0^+)$;
    \item The transition relation $\Gamma$ is defined below:
\end{itemize}\vspace{3mm}
\[\scriptsize
    \begin{array}{ll}
    
    \left \{ ((e,t),\ell,(e,t')) \left |~ 
    \begin{array}{l}
    \ell\in\Uncontrollable\ 
    \text{and}\ \exists t_1,t_2\ \text{s.t.}\ (t,\yieldE,t_1)
    \in\Gamma^+,\\
    (t_1,\yieldC,t_2)\in\Gamma^+\ \text{and}\ 
    (t_2,\ell,t')\in\Gamma^+ 
    \end{array}
    \right . \right \}
    & \bigcup \\[.5cm]

    \left \{ ((e,t),\ell,(e,t')) \left |~ 
    \begin{array}{l}
    \ell\in\Uncontrollable,
    (t,\ell,t')\in \Gamma^+\ \text{and}\\
    \forall t_1. (t,\yieldE,t_1)\in \Gamma^+ \rightarrow \Gamma^+_{\yieldC}(c,t_1)=\emptyset
    \end{array}
    \right . \right \}
    & \bigcup \\[.5cm]

     \left\{ ((e,t),\ell,(c,t')) ~\left |~ 
            \begin{array}{l}
                (t,\yieldE,t'')\in
                \Gamma^+\ \text{and}\\[4pt]
                (t'',\ell,t')\in
                \Gamma^+\ \text{and}\ \ell\in\Controllable
            \end{array}
        \right . \right\} & \bigcup \\[.5cm]

     \{ ((c,t),\ell,(c,t')) ~|~ \ell\in\Controllable\ 
        \text{and}\ (t,\ell,t')
        \in\Gamma^+\} & \bigcup \\[.2cm]
    
         \left\{ ((c,t),\ell,(e,t')) ~\left |~ 
                \begin{array}{l}
                    (t,\yieldC,t'')\in\Gamma^+\
                    \text{and}\\[4pt]
                    (t'',\ell,t')\in\Gamma^+\
                    \text{and}\ \ell\in\Uncontrollable
                \end{array}
            \right .\right \} &  \\[8pt]

    \end{array}
    \]
\end{proof}
Note that the states of $M$ retain as an extra memory the information of whether a state is on the ``environment side'' or the ``controller side''.  Intuitively, for a state $t$ of $M^+$ the controller $M$ adds the memory of whether $t$ was reached by a controllable or uncontrollable transition.
If $t$ is reached by a controllable transition, $M$ implements all transitions possible 
from $t$ and all transitions possible from the $\yieldC$ successor of $t$ (if exists).
Dually, if $t$ is reached by an uncontrollable transition, $M$ implements all transitions possible from $t$ and all transitions possible from the $\yieldE$ successor of $t$ (if exists). 
Furthermore, whenever possible add a detour that includes both a $\yieldE$ and a $\yieldC$ before an uncontrollable action. 

\begin{theorem}\label{thm:analysis-gr1-control}\emph{(Analysis \sgr Control)}
Let $E$ be a DLTS $E=(S, P, \alphabet, \Delta,  L, s_{0})$,  
where $\alphabet=\Uncontrollable \cupplus \Controllable$ is the set of actions controlled and monitored by $E$, respectively, and let $\varphi=\G \invariant \wedge (\bigwedge_{i=1}^n \G\F \assume_i \rightarrow \bigwedge_{j=1}^m \G\F \guarantee_j)$ 
be an FLTL formula with $\invariant$, $\assume_i$ and $\guarantee_j$ Boolean combinations of fluents.

A solution for the RTC control problem $\varepsilon$ exists if and only if the standard control problem $\langle\livelockRemove{E\|Y},\varphi', \Controllable\cup\{\yieldC\}\rangle$ is controllable, where $\varphi'$ is as follows.
\[\begin{array}{c}
\G\invariant \wedge
\G\F(\enenv \vee (\bigwedge_{\ell \in \Uncontrollable}\neg \ell^p_E) )\wedge\hfill\\[1ex]

([\G\F (\ensys \vee  (\bigwedge_{\ell \in \Controllable}\neg \ell^p_E)) \wedge \bigwedge_{i=1}^n \G\F \assume_i \wedge \G\F \alphabet]\\[1ex]  \rightarrow \bigwedge_{j=1}^m \G\F \guarantee_j)\\[1ex]
\end{array}\]
\end{theorem}

\begin{proof}
The only difference is in the location of the safety. 
    Every computation of $M\| live(E\| Y)$ is a computation of $M\|E$ interspersed with yield actions. It follows that $\G\invariant$ holds.
\end{proof}

\begin{corollary}
The complexity of RTC control with GR(1) goals is in $O(n\times m\times|S|^3)$, where $|S|$ is the number of states of the environment.
\end{corollary}

\begin{proof}
The goal in the modified control problem $\epsilon^+$ is of the form $\G\rho \wedge \G\F b \wedge (\bigwedge_{i=1}^n \G\F a_i \rightarrow \bigwedge_{j=1}^m \G\F g_j)$. 
By adding counters that range over the number of assumptions and the number of guarantees, this kind of goal can be converted to a Streett condition of index 2 \cite{streettcondition1982}. 
Control problems where the goal is a Streett condition of index 2 can be solved in time cubic in the number of states \cite{piterman2006faster}.  
\end{proof}

%% file: tex/relatedWork.tex

\section{Concluding Remarks and Future Directions}\label{section:relatedwork}
We introduced a novel control problem, named \emph{run-to-completion} (RTC), to deal with the asymmetric interaction between the controller and the environment commonly found in DES control. We showed that RTC control can be exploited to synthesise controllers for systems that can initiate and perform sequences of actions while responding correctly to external stimuli. 
This makes RTC suitable to control componentized systems with complex structures, and where a response to a single external stimulus may require several rounds of propagations among subsystems. Thanks to the flexible deadlines in RTC control, we are no longer required to count (or hardcode) the number of computation steps  for the system before it is ready again to react to the next stimulus. Furthermore, we avoid generating trivial controllers and we simplify the specifications by removing the explicit dependencies among subgoals, and thus facilitating incremental development. 

We showed that every instance of the RTC control problem can be reduced to a standard control problem, and finally we showed that when \sgr goals are used, RTC control can be reduced to Streett control of index 2~\cite{streettcondition1982}.

The notion of non-Zeno we have used is strongly related to fairness of the environment and the controller.
One could consider extensions in two different directions.
Our notion of non-Zeno allows the controller to force the environment to take some action. That is, in some cases where both controllable and uncontrollable actions are possible, we allow the controller to force the environment to move.
One could consider a weaker notion of non-Zenoness where the environment is not forced to take actions if it does not wish to do so.
Dually, we consider the controller non-Zeno if it often enough gives the environment the option to act (even if the environment cannot act).
This corresponds to the notion of weak fairness.
One could consider stronger restrictions on controllers in which they would have to be strongly-fair towards the environment. That is, if the environment can act infinitely often it should act infinitely often. 
Interestingly, one could consider even stronger notions, where strongly-fair controllers in addition completely block the environment only in cases where it is impossible to fulfil the goals when the environment acts infinitely often.

In many cases, studies of control of discrete event systems consider goals that are combinations of safety and non-blocking, while we have considered linear temporal goals.
In general, the techniques required to solve the two types of problems are very similar \cite{Ehlers17}.
The techniques developed in this paper can be adapted also to handle the case of non-blocking.
In the case of linear temporal goals it is well known that maximally permissive controllers do not exist. 
It is an interesting question whether RTC-control with safety and non-blocking goals allow for maximally permissive controllers.


\subsection{Related works}


In the prominent approach to  synthesis (such as \emph{Reactive Synthesis}~\cite{Pnueli:1989:SRM:75277.75293} and \emph{Supervisory Control}~\cite{ramadge89}) 
the uncontrolled plant has an advantage over the controller with respect to scheduling. Although, this may seem a natural understanding of the synthesis problem, where the controller is supposed to react to every possible behaviour of the plant, it is not always an appropriate assumption. In many cases, the uncontrolled plant is not completely adversarial (see~\cite{EhlersKB15}), and many undesired behaviours are practically infeasible and should be ruled out by definition, i.e., due to physical and/or software restrictions. For instance, a robot moving in an arena is restricted by its fixed structure~\cite{WongK15}, and thus it does not make sense to consider all possible paths between two points. 

These restrictions are usually dealt with by introducing domain specific assumptions over the plant (see~\cite{Kress-GazitFP09,piterman06}). However, these assumptions are not usually obvious and in many cases lead to spurious solutions (see~\cite{dippolito13}). 

A classic domain in which the uncontrolled plant is not completely adversarial is that of embedded systems where reactive languages (e.g.,~\cite{Halbwachs10, DBLP:conf/ifip/Berry89, Benveniste03}) adopt a synchronous hypothesis where the system can react to an external stimulus with all the computation steps it needs~\cite{DBLP:reference/crc/SimoneTP05}. To the best of our knowledge, RTC control is the first to automatically handle such assumptions.  

There have been many studies that focus on relating supervisory control and reactive synthesis, see~\cite{SchmuckMM20,Ehlers17}. However, some aspects are still not considered and that become more apparent with the approach presented herein. 
RTC control introduces a turn-based interaction between the controller and the plant that is similar to that of Reactive Synthesis~\cite{Pnueli:1989:SRM:75277.75293} for state-based models (i.e., no transition labels, only state propositions). Furthermore, in Reactive Synthesis both the controller and its adversary may perform in their own turn multiple actions concurrently. 
Yet in Reactive Synthesis the upper bound on actions per turn is determined by the number of state propositions, which is defined manually by the specifier before synthesis and it is not obvious how to reason about the order of concurrent events in state-based modelling. This becomes very important when dealing with systems that are required to do several rounds of data or control propagations among their subparts in response to external stimuli. Indeed, we may not know a-priory how many rounds of propagations are required or the order of events happening during the propagation. Furthermore, restricting the order (or the interleaving) of concurrent events manually might largely impact on the performance of the system under consideration. This is because hardcoded-orderings may easily sequentialise concurrent events that can safely be executed in parallel. 
Clearly, the last scenario poses a problem for both state-based models and event-based ones. Namely, once the events that can be executed simultaneously are explicitly identified, the flooding of adversarial events from the environment becomes as problematic as the one of uncontrollable events in discrete event systems. 
   



%% file: tex/appendix.tex
\section{Appendix}
\label{section:appendix}

Before we provide the proof for Theorem~\ref{thm:analysis-control}, we first introduce Lemma~\ref{lem:yield-properites} that shows the composition properties of the Yield DLTS. 
\begin{lemma}[Composition properties of Yield DLTS]\label{lem:yield-properites}
Let $E' = (S\times \{c,e\}, P', A^+, \Delta',  L', (s_0,e))$ with $A^+=U^+\cup C^+$ be a DLTS obtained by composing a domain model $E = (S, P, A, \Delta,  L, s_0)$ with a Yield DLTS $Y$. The followings hold:
\begin{enumerate}
\item   for all $\ell\in U$ . $\Delta'_{\ell}(s,e)\neq\emptyset$ {\bf iff} $\Delta_{\ell}(s)\neq\emptyset$.\medskip
    
 \item  for all $\ell\in U$ . $\Delta'_{\ell}(s,e)=\{(s',e)~|~ s'\in \Delta_{\ell}(s)\}$.\medskip
    
 \item for all $\ell\in C$ . $\Delta'_{\ell}(s,c)\neq\emptyset$ {\bf iff} $\Delta_{\ell}(s)\neq\emptyset$.\medskip
    
 \item for all $\ell\in C$ . $\Delta'_{\ell}(s,c)=\{(s',c)~|~ s'\in \Delta_{\ell}(s)\}$.\medskip

 \item $\Delta'_{\gamma_E}(s,e)=\{(s,c)\}$ {\bf and}  $\Delta'_{C^+}(s,c)=\Delta_{C}(s)\cup \Delta'_{\gamma_C}(s,c)$.\medskip

 \item $\Delta'_{\gamma_C}(s,c)=\{(s,e)\}$ {\bf and}  $\Delta'_{U^+}(s,e)=\Delta_{U}(s)\cup \Delta'_{\gamma_E}(s,e)$.\medskip

\end{enumerate}
\end{lemma}
\begin{proof}
    Clearly Lemma~\ref{lem:yield-properites} follows by the semantics of the parallel composition operator $\|$ in Definition~\ref{def:parallelcomp} and the construction of $Y$ in Definition~\ref{def:yieldLTS} respectively. 

\end{proof}

\begin{proof}[Theorem~\ref{thm:analysis-control}]

\begin{itemize}
\item[$\Rightarrow$]

Consider a solution $M$ for the RTC control problem $\varepsilon$. 
Let $M=(S\times T,P_M, A,\Gamma,L_M, t_0)$, where $\Gamma$ is the transition relation and $S\times T$ is the set of states of the controller. Note that since $M$ is a solution to $\varepsilon$ we consider the set of states to be $S\times T$ for some set $T$. This simplifies notations when considering the product with $E^+$. 
We define the DLTS $M^+=(S^+, P^+_M, A^+, \Gamma^+,L_M^+,s_0^+)$, where the components of $M^+$ are as follows.
\begin{itemize}
    \item 
    $S^+=S\times\{e\} \times T\ \bigcup\  S\times\{c\} \times T \times \{1,2\}$
    \footnote{Two copies of $S\times\{c\} \times T $ are needed to allow the controller to run to completion as othewise the environment can always schedule a $\yieldC$ immediately after  $\yieldE$ and the controller loses its chance to perform actual controllable actions. In the construction, from state $(s,c,t,1)$ you only yield when controllable actions are not possible.}
    \item 
    $A^+=A \ \bigcup\ \{\gamma_C, \gamma_E\}$
    \item
    $s_0^+=(s,e,t_0)$
    \item $\Gamma^+$ is defined as follows:\medskip
   
    \end{itemize}
    $$
    \begin{array}{l@{ } c@{\ } l@{\ } l@{\ }}
    \Gamma^+  =& & 
    \left \{ ((s,e,t),\ell,(s',e,t')) ~\left | \begin{array}{l} \ell \in \Uncontrollable \mbox{ and}\\  ((s,t),\ell,(s',t')) \in \Gamma \end{array} \right . \right \} \ \bigcup \\[3ex]
    
    & & \{ ((s,e,t),\yieldE,(s,c,t,1)) ~|~ \Delta_\Controllable(s)\neq \emptyset \}  \quad\bigcup \\[1ex]
    & & {\{ ((s,c,t,1),\yieldC,(s,e,t) ~|~ \Gamma_\Controllable(s,t)=\emptyset\}}\quad  \bigcup \\ 
    & & \left \{ ((s,c,t,1),\ell,(s',c,t',2)) ~\left | \begin{array}{l}\ell \in \Controllable \mbox{ and}\\\ ((s,t),\ell,(s',t')) \in \Gamma \end{array} \right. \right \}   \bigcup \\[3ex]
    & & \{ ((s,c,t,2),\yieldC,(s,e,t)) ~|~ \Gamma_\Uncontrollable(s,t)\neq \emptyset \} \quad \bigcup \\[1ex]
    & &\left \{ ((s,c,t,2),\ell,(s',c,t',2)) ~\left | \begin{array}{l}\Gamma_\Uncontrollable(s,t)=\emptyset \mbox{ and}\ \ell \in \Controllable \\ ((s,t),\ell,(s',t')) \in \Gamma \end{array} \right. \right \}
    \end{array}
    $$
Let the states of $E^+=\livelockRemove{E\|Y}$
be the pairs in $S\times \{e,c\}$ and $\Delta^+$ is the transition relation of $E^+$.        
We show that $M^+$ solves the control problem $\varepsilon^+$.
\begin{itemize}
    \item 
    $M^+$ is legal for $E^+$: By Definition~\ref{def:legalEnvironment} and the construction of $E^+$, it is sufficient to only consider reachable states of the form $((s,e), (s,e,t))\in E^+\| M^+$. Intuitively, these are the only states that can enable uncontrollable actions $\ell\in U^+$. Thus, $M^+$ is legal for $E^+$ if for every reachable state $((s,e), (s,e,t))\in E^+\| M^+$ and for all $\ell\in U^+$ such that $\Delta^+_{\ell}(s,e)\neq\emptyset$ we have $\Gamma^+_{\ell}(s,e,t)\neq\emptyset$. Note that the second item of Definition~\ref{def:legalEnvironment}  follows directly from the last item of Definition~\ref{def:rtclegalEnvironment}.\medskip
    
    Now by Lemma~\ref{lem:yield-properites}, we have that  for all $\ell\in U$ . $\Delta^+_{\ell}(s,e)\neq\emptyset$ {\bf iff} $\Delta_{\ell}(s)\neq\emptyset$. By assumption $M$ is RTC legal for $E$, and thus for every reachable state $(s,(s,t))\in E\| M$ if $\Gamma_{\Uncontrollable}(s,t)\neq\emptyset$  then also for every $\ell\in \Uncontrollable$ such that $\Delta_{\ell}(s)\neq \emptyset$ we have that $\Gamma_{\ell}(s,t)\neq \emptyset$. By the construction of $\Gamma^+$, we have also $\Gamma^+_{\ell}(s,e,t)\neq\emptyset$. Now it is clear that for every reachable state $(s,(s,t))\in E\| M$ where $\Gamma_{\Uncontrollable}(s,t)\neq\emptyset$  then there is a correspondent state $((s,e), (s,e,t))\in E^+\| M^+$ that mimics it on  every $\ell\in \Uncontrollable$. Furthermore, $(s,e)\in E^+$ implements a yield transition  $((s,e),\yieldE,(s,c))$ only when $\Delta_\Controllable(s)\neq \emptyset$ as established in Definition~\ref{def:removeLivelock}. This is exactly the case in $\Gamma^+$ where a yield transition $((s,e,t),\yieldE,(s, c, t,1))$ is only implemented when $\Delta_\Controllable(s)\neq \emptyset$.

    \item 
    $M^+\| E^+$ is deadlock free:
    As $M$ solves $\epsilon$ we know that $E\|M$ has no deadlocks.\medskip

    {\bf Case 1:} Consider a state $((s,e), (s,e,t))\in E^+\|M^+$. From the legality of $M^+$ and the construction of $\Gamma^+$ based on $M$, we easily conclude this case. Note that for every uncontrollable action $\ell\in U$ enabled in $M$ in state $(s,t)$, it must be enabled in $M^+$ in state $(s,e,t)$. If $\Delta_C(s)\neq\emptyset$ then $M^+$ in state $(s,e,t)$ must enable a yield transition $\gamma_E\in U^+$ and move to a state $(s,c,t,1)$.
     \medskip

    {\bf Case 2:} Consider a state $((s,c), (s,c,t,1))\in E^+\|M^+$. By Lemma~\ref{lem:yield-properites} and Definition~\ref{def:removeLivelock}, we can conclude that any controllable action $\ell\in C$ is enabled in $(s,c)\in E^+$ iff it is enabled in $s\in E$. Note that $E^+=live(E\| Y)$ only restricts $\yieldC, \yieldE$. By the construction of $\Gamma^+$ if $(s,t)$ enables some controllable action $\ell\in C$ then $(s, c, t,1)$ will implement it.  If $(s,t)$ enables some controllable action then by construction of $\Gamma^+$ we have that $M^+$ implements them from $(s,c,t,1)$.
    If $(s,t)$ enables no controllable action then $M^+$ implements the action $\yieldC$ leading from $(s,c,t,1)$ to $(s,e,t)$.  
\medskip
     
    {\bf Case 3:} Consider a state $((s,c), (s,c,t,2))\in E^+\|M^+$.
        If $(s,t)$ has no uncontrollable transition enabled from it, then there must be controllable transitions enabled from it.
    Then $(s,c,t,2)$ implements them. If $(s,t)$ has some uncontrollable action enabled from it, then the $\yieldC$ transition is enabled from $(s,c,t,2)$, and consequently $M^+$ moves to $(s,e,t)$ state where these actions are implemented. \medskip
    
    We conclude the deadlock-freedom of $M^+$.
    
    \item 
    $M^+\| E^+\models \varphi^+$:
    By assumption $M\|E \models \fairsys \wedge (\fairenv \rightarrow \varphi)$.
    Consider a computation $\pi^+$ of $M^+\| E^+$. Let $\pi$ be the corresponding computation in $M\| E$.

%
%
\medskip
    First, we show that $\pi^+$ satisfies $\G\F (\enenv \vee (\bigwedge_{\ell\in \Uncontrollable} \neg \ell^p_E))$. 
    
    By assumption we know that $\pi\models \fairsys$. That is $\pi$ satisfies either $\G\F u$ or $\G\F (Pass_E)$. If $\pi\models \G\F u$ then clearly $\pi^+ \models \G\F \enenv$. Whenever $\pi$ visits a state $(s,t)$ where $Pass_E$ holds then $\pi^+$ must visit a state $(s,c,t,i)$ where $\bigwedge_{\ell\in\Uncontrollable}\neg \ell^p_E$ holds. It follows that in the latter case $\pi^+$ satisfies $\G\F(\bigwedge_{\ell\in\Uncontrollable}\neg \ell^p_E)$.
    \medskip
   
   If $\pi^+ \not\models \G\F(\ensys \vee (\bigwedge_{\ell \in \Controllable} \neg \ell^p_E))$ then we are done. 
   Otherwise, assume that $\pi^+ \models \G\F(\ensys \vee (\bigwedge_{\ell \in \Controllable} \neg \ell^p_E))$.  We have to show that $\pi^+ \models \varphi$.

The difficulty is that in all locations of the form $(s,e,t)$ we have that $\pi^+$ satisfies $\bigwedge_{\ell \in \Controllable} \neg \ell^p_E$. However, it is not necessarily the case that $(s,t)$ satisfies $\bigwedge_{\ell \in \Controllable} \neg \ell^p_M$.

   We construct from $\pi^+$ a computation $\pi^{++}$ that differs from $\pi^+$ only by adding $\yieldC$ and $\yieldE$ actions. Thus, $\pi^{++} \models \varphi$ if and only if $\pi^+$ satisfies $\varphi$. However, $\pi^{++}$ satisfies $\G\F (\bigwedge_{\ell \in \Controllable} \neg \ell^p_E) \wedge \F\G \neg \ensys$ if and only if $\pi$ satisfies $\G\F (\bigwedge_{\ell \in \Controllable} \neg \ell^p_M)$.
   
   Consider the computation $\pi^+$. We obtain $\pi^{++}$ from $\pi^+$ by replacing every occurrence of a state $(s,e,t)$ where $\Delta_C(s)\neq \emptyset$ but $\Gamma_C(s,t)=\emptyset$ by the sequence $(s,e,t),(s,c,t,1),(s,e,t)$ resulting from adding $\yieldE$ and $\yieldC$ transitions.
   
   Clearly, $\pi^{++}$ and $\pi^+$ agree on all actions in $A$ (i.e., all actions except $\yieldE$ and $\yieldC$).
   Thus, $\pi^{++}$, $\pi^+$ and $\pi$ all agree on satisfying or not satisfying $\varphi$. 
   Also, whenever $\pi^+$ visits a state where $(s,e,t)$ such that $(s,t) \not \models \bigwedge_{\ell \in \Controllable} \neg \ell^p_M$ we added a loop where the fluent $\ensys$ holds. 
   
   Now, if $\pi^{++} \models \G\F \ensys$ then it is either the case that $\pi^{++}$ satisfies $\G\F(\bigwedge_{\ell \in \Controllable} \neg \ell^p_M)$ or $\pi^{++}$ satisfies $\G\F c$. In the first case, $\pi$ satisfies $\G\F Pass_M$. In the second case, $\pi$ satisfies $\G\F c$. The only remaining case is where $\pi^{++}$ visits infinitely many locations $(s,e,t)$ such that $\Delta_C(s)=\emptyset$ but $\pi^{++}\models \F\G\neg \ensys$. However, whenever in $(s,e,t)$ we have $\Delta_c(s)=\emptyset$ we know from legality of $M$ that $\Gamma_C(s,t)=\emptyset$. It follows that in this case $\pi$ satisfies $\G\F(\bigwedge_{\ell \in \Controllable} \neg \ell^p_E)$ in such cases, from legality of $M$ it follows that $(s,t)$ satisfies satisfies $\G\F (\bigwedge_{\ell \in \Controllable} \neg \ell^p_M)$ as well. 
   We conclude that $\pi$ satisfies $\G\F (c \vee Pass_M)$. 
   
   As $M$ is an RTC controller for $E$, we have $\pi\models \varphi$. 
   So $\pi^{++}$ satisfies $\varphi$ as well and hence also $\pi^+$. 

\end{itemize}
\item[$\Leftarrow$]
Consider a solution $M^+$ for the control problem $\varepsilon^+$. 
To simplify notations we consider the set of states of $M^+$ to be $S \times \set{e,c} \times T$ for some $T$. As $M^+$ is a solution to $\varepsilon^+$ this is equivalent to considering what happens in the product of $M^+$ and $E^+$. 

Let $M^+=(T^+, P^+, A^+, \Gamma^+,L^+,t_0^+)$, where $T^+=S \times \set{e,c}\times T$ is the set of states of the controller, $\Gamma^+$ is the transition relation, and  the initial state $t_0^+=(s_0,e,t_0)$. 

Let $M=(T', P_M, A, \Gamma,L_M,(s_0,e,t_0))$ be the DLTS with the following components.
\begin{itemize}
    \item 
    $T'=S\times \set{e,c} \times T$ is the set of states (i.e., same as $M^+$)
    \item 
    $A=A^+ \ \backslash\ \set{\yieldC, \yieldE}$ is the alphabet
    \item
    $(s_0,e,t_0)$ is the initial state
    \item 

$$
    \begin{array}{l c l r}
    \Gamma & = & 
    \left \{ ((s,e,t),\ell,(s',e,t')) \left |~ 
    \begin{array}{c} \ell\in\Uncontrollable\ 
    \text{and}\ \exists t'',t''' \text{ s.t. } 
    ((s,e,t),\yieldE,(s,c,t''))\in\Gamma^+,\\
    ((s,c,t''),\yieldC,(s,e,t''')) \in \Gamma^+, \text{ and }
    ((s,e,t''),\ell,(s',e,t'))
    \in\Gamma^+\end{array} \right. \right \} & \bigcup \\[.5cm]

& &    \left \{ ((s,e,t),\ell,(s',e,t')) \left |~ 
    \begin{array}{c} \ell\in\Uncontrollable,
    ((s,e,t),\ell,(s',e,t'))
    \in\Gamma^+, \mbox{ and}\\
    \forall t'' ~.~ 
    ((s,e,t),\yieldE,(s,c,t''))\in\Gamma^+ \rightarrow
    \Gamma^+_{\yieldC}(s,c,t'')=\emptyset
    \end{array} \right. \right \} & \bigcup \\[.5cm]

    & & \left\{ ((s,e,t),\ell,(s',c,t')) ~\left |~ 
            \begin{array}{c}
                ((s,e,t),\yieldE,(s,c,t''))\in
                \Gamma^+\ \text{and}\\[4pt]
                ((s,c,t''),\ell,(s',c,t'))\in
                \Gamma^+\ \text{and}\ \ell\in\Controllable
            \end{array}
        \right . \right\} & \bigcup \\[.5cm]

    & & \{ ((s,c,t),\ell,(s',c,t')) ~|~ \ell\in\Controllable\ 
        \text{and}\ ((s,c,t),\ell,(s',c,t'))
        \in\Gamma^+\} & \bigcup \\[.5cm]
    
        & & \left\{ ((s,c,t),\ell,(s',e,t')) ~\left |~ 
                \begin{array}{c}
                    ((s,c,t),\yieldC,(s,e,t''))\in\Gamma^+\
                    \text{and}\\[4pt]
                    ((s,e,t''),\ell,(s',e,t'))\in\Gamma^+\
                    \text{and}\ \ell\in\Uncontrollable
                \end{array}
            \right .\right \} &  \\[8pt]

    \end{array}
    $$    
\end{itemize}
Notice that states of $M$ retain as an extra memory the information of whether a state is on the ``environment side'' or the ``controller side''.  
Furthermore, if it is possible to ``do a round'' through the controller side before taking an uncontrollable action, then this is taken. 

We show that $M$ solves the RTC control problem $\varepsilon$.
\begin{itemize}
    \item 
    $M$ is RTC legal for $E$: By definition~\ref{def:rtclegalEnvironment}, we show that for every reachable state $(s,u)$ of $E\| M$ the following holds.
\begin{enumerate}
\item 
If $\Gamma_{\Uncontrollable}(u)\neq\emptyset$  then for every $\ell\in \Uncontrollable$ such that $\Delta_{\ell}(s)\neq \emptyset$ we have that $\Gamma_{\ell}(u)\neq \emptyset$.\medskip

If $u=(s,e,t)$ then as $(s,e,t)$ implements all uncontrollable transitions in $M^+$ the same is true for $u$ in $M$. Either by taking the transition directly or taking a detour of $\yieldE$ and $\yieldC$ before taking the transition. 

If $u=(s,c,t)$ there are two options. Either there is no transition $((s,c,t),\yieldC,(s,e,t'))$ in $M^+$, in which case no uncontrollable transitions are implemented in $\Gamma$. Or there is a transition $((s,c,t),\yieldC,(s,e,t'))$ in $M^+$, in which case all uncontrollable transitions possible from $s$ are implemented in $\Gamma$.
\item
 If $u\in \Gamma_{\Uncontrollable}(T)$ then for every $\ell\in \Uncontrollable$ such that $\Delta_{\ell}(s)\neq \emptyset$ we have that $\Gamma_\ell(u)\neq \emptyset$.\medskip
 
 By construction, we know that $u=(s,e,t)$. From legality of $M^+$ we conclude that for every $\ell \in \Uncontrollable$ such that $\Delta_{\ell}(s)\neq \emptyset$ we have $\Gamma_\ell(s,e,t)\neq \emptyset$.
\item
For every $\ell\in\Controllable$ such that $\Delta_\ell(s)=\emptyset$ we  have that $\Gamma_\ell(t)=\emptyset$.\medskip

This item follows directly from the fact that $\Gamma$ is based on $\Gamma^+$, and thus only enables controllable actions that are enabled in $\Gamma^+$. Furthermore,  $M^+$ is legal for $E^+$, and thus $\Gamma^+$ does not enable controllable actions that are  not enabled in $\Delta^+$, and consequently are not enabled in $\Delta$ as required. 
\end{enumerate}

    \item 
    $M\| E$ is deadlock free: As $M^+$ solves $\epsilon^+$ we know that $E^+\|M^+$ has no deadlocks.\medskip
    
    Consider a state $(s, (s,c,t))\in E\|M$. 
    If $\Gamma^+_{\yieldC}(s,c,t)\neq\emptyset$ then, by definition of $\livelockRemove{E\| Y}$ we know that $\Gamma^+_{\Uncontrollable}(s,e,t)\neq\emptyset$. It follows that $(s,c,t)$ implements all uncontrollalble transitions implemented from $s$.
    If $\Gamma^+_{\yieldC}(s,c,t)=\emptyset$ then, by $((s,c),(s,c,t))$ not being a deadlock in $M^+$ we know that $\Gamma^+_{\Controllable}(s,c,t)\neq \emptyset$.
    
    Consider a state $(s,(s,e,t)) \in E\|M$.
    We know that $((s,e),(s,e,t))$ is not a deadlock in $E^+\| M^+$.
    If some uncontrollable transition is available from $s$ then it is implemented from $(s,e,t)$.
    If no uncontrollable transition is available from $s$ then it must be the case that $\yieldE$ is available from $((s,e),(s,e,t))$ to $((s,c),(s,c,t')$. However, as no uncontrollable transition is available from $S$ it follows that $\yieldC$ is not implemented from $(s,c)$ implying that some controllable transition is available from $((s,c),(s,c,t'))$.
    It follows that the same controllable transition is available from $(s,(s,e,t))$.
    \newpage

    \item 
    $M\| E\models \varphi$:
        By assumption $M^+\|E^+ \models \varphi^+$.
        Consider a computation $\pi$ of $M\| E$.
        Every transition in $\pi$ corresponds to a transition or a sequence of transitions in $M^+\|E^+$.
        Consider the computation $\pi^+$ that is the concatenation of all these transitions.
        
    First, we show that $\pi\models \fairsys$. That is $\pi$ satisfies either $\G\F u$ or $\G\F Pass_E$.
    By assumption we know that $\pi^+$ satisfies $\G\F (\enenv \vee (\bigwedge_{\ell\in \Uncontrollable} \neg \ell^p_E))$. 

    Whenever $\pi^+$ satisfies $\bigwedge_{\ell\in\Uncontrollable} \neg \ell^p_E$ so does $\pi$.
    Consider the case that $\pi^+$ satisfies $\G\F \enenv$.
    
    By the structure of $E^+$, $\pi^+$ satisfies $\enenv$ either after a transition labeled by some $\ell \in \Uncontrollable$ has been taken or after a transition labeled $\yieldC$.
    However, by construction of $M$, a transition in $\pi^+$ labeled by $\yieldC$ appears only as part of a pair of transitions corresponding to a single transition labeled by $\ell\in\Uncontrollable$ of $\pi$. 
    It follows that $\pi\models \G\F u$.
    \medskip
    
    Now, as $M^+$ is a solution for $\epsilon^+$, $\pi^+$ satisfies $(\G\F A \wedge \G\F (\ensys \vee (\bigwedge_{\ell \in\Controllable}\neg \ell^p_E))) \rightarrow \varphi$. We have to show that $\pi$ satisfies $(\psi_c \rightarrow \varphi)$. \medskip
      
    If $\pi\not\models \psi_c$ then we are done.
    Otherwise, consider the case that $\pi$ satisfies $\psi_c$.
      
    If $\pi$ satisfies $\G\F c$ then $\pi^+$ satisfies $\G\F A$ and $\G\F\ensys$ implying that both $\pi^+$ and $\pi$ satisfy $\varphi$.

    Otherwise, $\pi$ satisfies $\G\F pass_M\equiv \G\F (\bigwedge_{\ell\in\Controllable}\neg\ell^p_M)$. 
    Consider a state $(s,(s,e,t))$ of $\pi$ where $\bigwedge_{\ell\in\Controllable}\neg\ell^p_M$ is true. This is matched by $\pi^+$ visiting state $(s,e,(s,e,t))$ and the next action to happen in $\pi$ must be some $\ell\in\Uncontrollable$.
    One of the following holds:
    \begin{itemize}
        \item 
        either $(s,(s,e,t))$ satisfies $\bigwedge_{\ell\in\Controllable}\neg\ell^p_E$ as well.  
        
        In this case the $\yieldE$ transition is not included from $(s,e)$ in $\livelockRemove{E\|Y}$.
        Furthermore, $\pi^+$ also satisfies $\bigwedge_{\ell\in\Controllable}\neg\ell^p_E$ in state $(s,e,(s,e,t))$, as this depends only on the state $s$ of $E$.
        \item 
        or $(s,(s,e,t))$ does not satisfy $\bigwedge_{\ell\in\Controllable}\neg\ell^p_E$.
        
        In this case, $(s,e)$ has a $\yieldE$ successor $(s,c)$ in $\livelockRemove{E\|Y}$. 
        Let $(s,c,(s,c,t''))$ be the $\yieldE$ successor of $(s,e,(s,e,t))$ in $M^+\|E^+$.
        We know that $(s,(s,e,t))$ satisfies $\bigwedge_{\ell\in\Controllable}\neg\ell^p_M$.
        Thus, $(s,c,(s,c,t''))$ has no controllable transition possible from it (otherwise $(s,e,t)$ would include that successor contradicting $(s,(s,e,t))$ satisfying $pass_M$).
        However, $(s,c,(s,c,t''))$ cannot be a deadlock.
        Then, $(s,c,(s,c,t''))$ has a $\yieldE$ successor $(s,e,(s,c,t'''))$.
        By construction, the (sequence of three) transitions in $\pi^+$ that correspond to the transition in $\pi$ on action $\ell$ includes a state where $\ensys$ is true.
    \end{itemize}
    Overall, as $\pi$ satisfies $\G\F pass_M$ it follows that $\pi^+$ satisfies either $\G\F\ensys$ or $\G\F(\bigwedge_{\ell\in\Controllable}\neg\ell^p_E)$.
    Furthermore, as $\pi$ includes infinitely many actions, $\pi^+$ must satisfy $\G\F A$.

    Thus, both $\pi^+$ and $\pi$ satisfy $\varphi$.
      
\end{itemize}
    We note that this construction maintains determinism of the generated RTC controller. 
    That is, if the controller $M^+$ is deterministic then the controller $M$ is deterministic as well.
\end{itemize}
\end{proof}